\documentclass[a4paper, 12pt]{article}
\usepackage{bm}
\usepackage{amsmath}

\usepackage{natbib}

\usepackage{booktabs}
\usepackage{threeparttable}
\usepackage[dvipdfmx]{graphicx}
\usepackage{xcolor}

\usepackage[hang,small,bf]{caption}
\usepackage[subrefformat=parens]{subcaption}
\usepackage{amsthm}
\theoremstyle{plain}
\newtheorem{prop}{Proposition}
\theoremstyle{remark}

\begin{document}

\def\spacingset#1{\renewcommand{\baselinestretch}%
{#1}\small\normalsize} \spacingset{1}

\title{A regression approach to the two-dataset problem}
\author{
Steven N. MacEachern \\ Department of Statistics, The Ohio State University
\and
Koji Miyawaki \\ School of Economics, Kwansei Gakuin University
}
\date{\today}
\maketitle

\begin{abstract}
This paper considers the two-dataset problem, where data are collected from two potentially different populations sharing common aspects.
This problem arises when data are collected by two different types of researchers or from two different sources.
We may reach invalid conclusions without using knowledge about the data collection process.
To address this problem, this paper develops statistical regression models focusing on the difference in measurement and proposes two prediction errors that help to evaluate the underlying data collection process.
As a consequence, it is possible to discuss the heterogeneity/similarity of the set of predictors in terms of prediction.
Two real datasets are selected to illustrate our method.

\textit{Keywords: } Data collection process; Random coefficients model; Bayesian model averaging.
\end{abstract}

\newpage
\spacingset{1.5} 


\section{Introduction}


The data collection process is an important component of statistical analysis (see, e.g., \citet{cox-snell-81}), where data are often assumed to be independently collected from an identical population with appropriate precision for the analysis.
In the literature on causal inference, this assumption is required for so-called internal validity (see, e.g., \citet{shadish-etal-02}).
However, in practice, we encounter cases where this assumption does not hold because of the data collection process.
This paper develops statistical regression models that reflect the underlying data collection process focusing on the difference in measurement and proposes two prediction errors that can be used to evaluate the process as well as the heterogeneity/similarity in the set of predictors.
These two errors are of Mallows' $C_{p}$-type (see \citet{mallows-73, mallows-95}) and closely related to the idea of the final prediction error (see \citet{akaike-69, akaike-70}), which is a precedent of the well-known Akaike information criterion.

Two typical data collection processes, neither of which satisfies the above assumption, are considered in this paper.
The first situation concerns data quality.
Except in cases of small datasets, data are typically collected by several researchers and are combined into one large dataset for analysis.
Thus, it is reasonable to consider the heterogeneity in the degree of data collection skill among researchers.
Even if the skill is homogeneous, the data may differ in quality due to the data collection environment, as we will see in the geyser dataset (Figure \ref{fig:The Old Faithful Geyser data.}).
Both scenarios yield a dataset that is a mix of low-quality and high-quality data.
In this case, internal consistency in quality (or data compatibility) is an issue.
If the data are not consistent in quality, using only the precise data may lead to a valid conclusion, but the conclusion may be inaccurate because it ignores information from the discarded data.

The second situation is related to the data source.
To ensure a sufficient number of observations, data are collected from multiple sources, which may violate the assumption that the data are from the same population.
For example, medical data are often collected at several hospitals, but different hospitals are likely to have different treatments or patients.
Thus, the analysis of data collected in this way may be invalid because the data are from different populations.
However, as in the diabetes dataset discussed in Section \ref{sec:Motivating data}, such different populations also share some common aspects.
In practice, data are obtained by a mix of these two situations mentioned above.

To address above two situations explicitly, we link the process, focusing on the difference in measurement, to the linear regression model, and derive two prediction errors, which are estimated by the Bayesian approach.
The first evaluates the internal data compatibility.
Furthermore, we investigate the heterogeneity/similarity in the set of predictors with the second error, which is addressed within the framework of Bayesian model averaging (see, e.g., \citet{raftery-etal-97, brown-etal-02, steel-19, miyawaki-maceachern-19}).

Situations similar to the one considered in this paper are addressed in the classical estimation methods, including the generalized $T^{2}$-statistics applied in the two-sample problem (see \citet{hotelling-31}), the canonical correlation analysis applied in the test of independence (see \citet{hotelling-36}), and the test of mean vectors of different distributions using the multivariate analysis of variance (see, e.g., \citet{anderson-03}).
More recently, statistical inferences for big data are demanded, and the distributed computing plays an important role in the literature (see \citet{gao-etal-21} for its review), the technology of which is related to the problem examined in this paper.

The rest of this paper is organized as follows.
After two real datasets used in this paper are introduced in Section \ref{sec:Motivating data}, Section \ref{sec:Two-dataset problem} describes a general two-dataset problem and derives two prediction errors, focusing on the difference in measurement.
Their Bayesian estimation is explained in Section \ref{sec:Estimate the prediction error}.
The sensitivity of our estimation approach to the prior distribution is also discussed in this section.
Section \ref{sec:Illustrative applications} presents illustrative applications of our method using two real datasets.
Finally, Section \ref{sec:Conclusion} concludes the paper.


\section{Two motivating datasets}
\label{sec:Motivating data}


The geyser data used by \citet{azzalini-bowman-90} are shown in Figure \ref{fig:The Old Faithful Geyser data.}.
\begin{figure}[ht]
\centering
\includegraphics[width=10cm, clip, keepaspectratio]{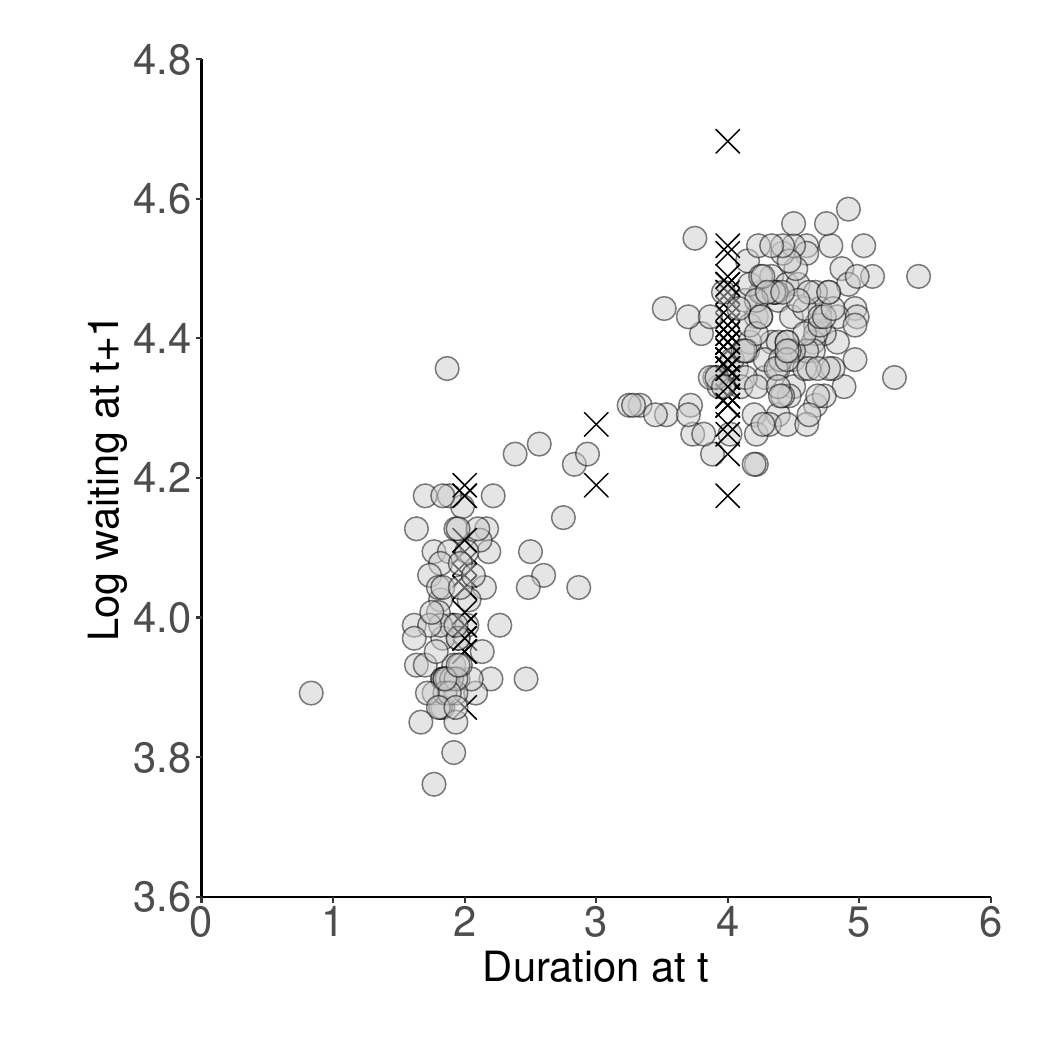}
\caption{The Old Faithful Geyser data.}
\label{fig:The Old Faithful Geyser data.}
\end{figure}
Each observation is the repeated measurement of eruptions of the Old Faithful Geyser in the United States, and consists of two variables: the duration of the eruption and the waiting time between two consecutive eruptions.
Of the 298 observations, 77 have duration times of 2, 3, or 4, as depicted by the x marks in Figure \ref{fig:The Old Faithful Geyser data.}.
These values are rounded because they are collected at night, as noted in \citet{azzalini-bowman-90}.
This is an example of heterogeneity in data quality.

Another example is the diabetes data analyzed by \citet{efron-etal-04}, including diabetes progression measure (Y) as the response and ten predictors (age, sex, body mass index (BMI), average blood pressure (BP), and six blood serum measurements (S1 - S6)).
This dataset is used to illustrate variable selection methods (e.g., \citet{efron-etal-04}, \citet{hahn-carvalho-15}, \citet{miyawaki-maceachern-19} among others).
A close look suggests that the data may come from at least two sources, because the precision of the blood pressure and S4 (the fourth blood serum measurement) is different from that of the other measurements (see ID 95 and 99 in Table \ref{table:Selected 10 rows of diabetes data}).
\begin{table}[!ht]
\begin{center}
\catcode`?=\active \def?{\phantom{0}}
\catcode`!=\active \def!{\phantom{.}}
\caption{Ten observations selected from the diabetes data}
\label{table:Selected 10 rows of diabetes data}
\scalebox{0.9}{
\begin{threeparttable}
\begin{tabular}{@{}rccccccccccc}
\toprule
ID\tnote{$\dagger$} & Age & Sex & BMI & BP & S1 & S2 & S3 & S4 & S5 & S6 & Y \\
                    & ($x_{1}$) & ($x_{2}$) & ($x_{3}$) & ($x_{4}$) & ($x_{5}$) & ($x_{6}$) & ($x_{7}$) & ($x_{8}$) & ($x_{9}$) & ($x_{10}$) & \\
\midrule
91  & 36 & 1 & 21.9 & ?89!?? & 189 & 105.2 & 68 & 3!?? & 4.3694 & ?96 & 111 \\
92  & 52 & 1 & 24!? & ?83!?? & 167 & ?86.6 & 71 & 2!?? & 3.8501 & ?94 & ?98 \\
93  & 61 & 1 & 31.2 & ?79!?? & 235 & 156.8 & 47 & 5!?? & 5.0499 & ?96 & 164 \\
94  & 43 & 1 & 26.8 & 123!?? & 193 & 102.2 & 67 & 3!?? & 4.7791 & ?94 & ?48 \\
95  & 35 & 1 & 20.4 & ?65!?? & 187 & 105.6 & 67 & 2.79 & 4.2767 & ?78 & ?96 \\
96  & 27 & 1 & 24.8 & ?91!?? & 189 & 106.8 & 69 & 3!?? & 4.1897 & ?69 & ?90 \\
97  & 29 & 1 & 21!? & ?71!?? & 156 & ?97!? & 38 & 4!?? & 4.654? & ?90 & 162 \\
98  & 64 & 2 & 27.3 & 109!?? & 186 & 107.6 & 38 & 5!?? & 5.3083 & ?99 & 150 \\
99  & 41 & 1 & 34.6 & ?87.33 & 205 & 142.6 & 41 & 5!?? & 4.6728 & 110 & 279 \\
100 & 49 & 2 & 25.9 & ?91!?? & 178 & 106.6 & 52 & 3!?? & 4.5747 & ?75 & ?92 \\
\bottomrule
\end{tabular}
\begin{tablenotes}[normal]
\item[$\dagger$] The variable ID indicates the number of rows in the data file, which does not appear in the original.
\end{tablenotes}
\end{threeparttable}
}
\end{center}
\end{table}
Such observations are 65 out of 442.

As a preliminary analysis of the diabetes data, we separated observations into two sets based on precision and applied a multivariate linear regression model to each dataset.
The prior distribution was specified as noted in Section \ref{sec:Estimate the prediction error}.
The marginal likelihood was maximized with $(x_{2}, x_{3}, x_{9}, x_{10})$ when only precise data were used and with $(x_{2}, x_{3}, x_{4}, x_{7}, x_{9})$ when only imprecise data were used.
The results suggest that the data may be obtained from two different populations, keeping in mind that the two populations also have common aspects: they are collected from diabetes patients.
How the populations are different/similar will be investigated in Subsection \ref{subsec:Similarity search}.


\section{Two-dataset problem}
\label{sec:Two-dataset problem}


Suppose $n$ observations are independently drawn from one of two populations.
These two populations may be the same or different.
We know which observation comes from which population in terms of, say, precision, but have little knowledge of whether the two populations are identical.
In this paper, this is called the two-dataset problem.
More mathematically, one part of the dataset is randomly sampled from the distribution $f_{1}$, while the other part is from the distribution $f_{2}$.
Parameters of $f_{1}$ and $f_{2}$ maybe the same or different.

Because this paper considers the regression framework, we focus on the conditional expectation of the response given predictors.
The regression framework requires two aspects to consider.
One is the marginal effect on the conditional expectation, and the other is the selection of predictors.
Subsection \ref{subsec:Random coefficients model} addresses the former, while subsection \ref{subsec:A generalization of the random coefficients model} considers both aspects.

Let $\mathcal{S}_{0}$ and $\mathcal{S}_{1}$ be sets of observation indices of samples drawn from the first and second populations, respectively.
Then, $\mathcal{S}_{0} \cup \mathcal{S}_{1} = \{ 1, \dots, n \}$ and $\mathcal{S}_{0} \cap \mathcal{S}_{1} = \emptyset$.
In the context of our application, we use $\mathcal{S}_{0}$ for low-quality data and $\mathcal{S}_{1}$ for high-quality data.
Specifically, $| \mathcal{S}_{0} | = 77$ and $| \mathcal{S}_{1} | = 221$ in the geyser dataset, while $| \mathcal{S}_{0} | = 377$ and $| \mathcal{S}_{1} | = 65$ in the diabetes dataset.
It is worth noting the analysis that takes care of outliers is not appropriate because it restricts ourselves to one identical population, which may not be true in the two-dataset problem.
Furthermore, both kinds contain substantial observations, which is not typical in the outlier analysis.
Thus, they require a different modeling while keeping in mind that they share common aspects.

We apply linear regression to this type of dataset to infer the underlying mechanism between the response and predictors.
If the assumption that the data are from the identical population is the state of nature, estimation methods with the entire dataset are asymptotically correct.
If the data are from different sources, pooled estimation is not appropriate, and estimation with separate datasets is preferable.
However, as in the diabetes dataset, the data may be from different populations that share some aspects to a certain degree.
To address this issue, we consider two aspects of the model: regression coefficients and a set of predictors.
When these factors are taken into account, we can obtain a criterion that sheds light on the two-dataset problem.

The general prediction error considered in this paper is the following.
Let $\hat{f}_{0} (\bm{x})$ be a prediction based on data for $\mathcal{S}_{0}$.
Then, the prediction error of data $\mathcal{S}_{1}$ is given by
\begin{align*}
E \left[ L \left\{ Y_{1} - \hat{f}_{0} \left( \bm{x}_{1} \right) \right\} \mid \left( y_{i}, \bm{x}_{i} \right)_{i \in \mathcal{S}_{0}}, \bm{x}_{1} \right],
\end{align*}
where $(Y_{1}, \bm{x}_{1})$ is the response and predictors obtained from the population for $\mathcal{S}_{1}$ and $L (\cdot)$ is the loss function.
In this paper, we use the squared loss.
This expression is the error of data $\mathcal{S}_{1}$ given $\mathcal{S}_{0}$.
Similarly, we define the error of data $\mathcal{S}_{0}$ given $\mathcal{S}_{1}$.
The following two subsections apply this framework to the linear regression model, in which the prediction error is specified by the model parameters and observed data.


\subsection{Random coefficients model}
\label{subsec:Random coefficients model}


A possible specification of the two-dataset problem is the random coefficients model (see, e.g., \citet{griffiths-etal-79}).
It assumes associations between the response and predictors that are different between datasets while coefficients are obtained from the same population.
Then, as specified below, a parameter in the random coefficients assumption demonstrates how these two associations are close with each other.
Hereafter, predictors are standardized within each dataset for brevity.

For each dataset, we consider a linear regression model with different regression coefficients, which is given by
\begin{align*}
\begin{cases}
Y_{i} = \bm{x}_{0i}^{\prime} \bm{\beta}_{0} + \epsilon_{0i}, &\text{for $i \in \mathcal{S}_{0}$}, \\
Y_{i} = \bm{x}_{1i}^{\prime} \bm{\beta}_{1} + \epsilon_{1i}, &\text{for $i \in \mathcal{S}_{1}$},
\end{cases}
\end{align*}
where $Y_{i}$ is the response for observation $i$ ($i = 1, \dots, n$), $\bm{x}_{ji}$ is a vector of predictors in model $j$, and the error term $\epsilon_{ji}$ has mean $0$ and variance $\sigma_{j}^{2}$ for $j = 0, 1$.
In this specification, the predictors are common and arranged in the same order across the two datasets.
When data for $\mathcal{S}_{j}$ are stacked by rows, the above model can be represented in matrix form as
\begin{align}
\begin{cases}
\bm{y}_{0} = \bm{X}_{0} \bm{\beta}_{0} + \bm{\epsilon}_{0}, \\
\bm{y}_{1} = \bm{X}_{1} \bm{\beta}_{1} + \bm{\epsilon}_{1},
\end{cases}
\label{eq:RCM common}
\end{align}
where $\bm{y}_{j} = (Y_{i})_{i \in \mathcal{S}_{j}}$, $\bm{X}_{j}^{\prime} = (\bm{x}_{ji})_{i \in \mathcal{S}_{j}}$, and $\bm{\epsilon}_{j} = (\epsilon_{ji})_{i \in \mathcal{S}_{j}}$ for $j = 0, 1$.

To exploit the similarity across datasets, variation in the regression coefficients is explicitly considered, so that the random coefficients assumption is introduced; that is,
\begin{align}
\bm{\beta}_{1} = \bm{\beta}_{0} + \bm{\eta},
\label{eq:random coefficients assumption}
\end{align}
where $\bm{\eta}$ has mean $0$ and variance covariance matrix $\sigma_{\eta}^{2} \bm{I}$ and $\bm{I}$ is the unit matrix.
Independence between $\bm{\eta}$ and $\bm{\epsilon}_{j}$ is assumed $(j = 0, 1)$.

The former error represents the deviance of the association between the response and predictors among datasets, while the latter includes the measurement and specification errors.
When the skill heterogeneity exists, it might be reasonable to assume dependence between these two terms because such heterogeneity might affect the association as well as the measurement.
However, because the data are from the identical population in this case, the deviance is expected to be relatively small, and the independence assumption would not much affect the final result, i.e., the prediction error.
In addition, it is possible to incorporate more structures in this assumption such as dependence and heteroskedasticity, depending on the context to apply.
However, to make things as simple as possible, we use this specification in this paper.

Under the above setting, the following proposition holds.
\begin{prop}
Let $\hat{\bm{y}}_{1 \mid 0}$ be the prediction of $\bm{y}_{1}$ using the least squares estimate of $\bm{\beta}_{0}$.
Then, the expectation of the average squared prediction error is given by
\begin{align*}
\frac{1}{n_{1}} E \left( \bm{y}_{1} - \hat{\bm{y}}_{1 \mid 0} \right)^{\prime} \left( \bm{y}_{1} - \hat{\bm{y}}_{1 \mid 0} \right)
=
\sigma_{1}^{2} + k \sigma_{\eta}^{2} + \frac{\sigma_{0}^{2}}{n_{1}} \text{tr} \left\{ \left( \bm{X}_{0}^{\prime} \bm{X}_{0} \right)^{-1} \bm{X}_{1}^{\prime} \bm{X}_{1} \right\},
\end{align*}
where $n_{j} = | \mathcal{S}_{j} |$ and $k$ is the number of predictors.
\label{prop:rcm}
\end{prop}
\begin{proof}
Because the least squares estimate of $\bm{\beta}_{0}$ is $( \bm{X}_{0}^{\prime} \bm{X}_{0} )^{-1} \bm{X}_{0}^{\prime} \bm{y}_{0}$, the prediction based on data for $\mathcal{S}_{0}$ is given by
\begin{align*}
\hat{\bm{y}}_{1 \mid 0} &= \bm{X}_{1} \left( \bm{X}_{0}^{\prime} \bm{X}_{0} \right)^{-1} \bm{X}_{0}^{\prime} \bm{y}_{0}
                         = \bm{X}_{1} \left( \bm{X}_{0}^{\prime} \bm{X}_{0} \right)^{-1} \bm{X}_{0}^{\prime} \left( \bm{X}_{0} \bm{\beta}_{0} + \bm{\epsilon}_{0} \right) \notag \\
                        &= \bm{X}_{1} \bm{\beta}_{1} - \bm{X}_{1} \bm{\eta} + \bm{X}_{1} \left( \bm{X}_{0}^{\prime} \bm{X}_{0} \right)^{-1} \bm{X}_{0}^{\prime} \bm{\epsilon}_{0}.
\end{align*}
The independence of error terms leads to
\begin{align*}
E \left( \bm{y}_{1} - \hat{\bm{y}}_{1 \mid 0} \right)^{\prime} \left( \bm{y}_{1} - \hat{\bm{y}}_{1 \mid 0} \right)
=
&E \left( \bm{\epsilon}_{1}^{\prime} \bm{\epsilon}_{1} \right)
+
E \left( \bm{\eta}^{\prime} \bm{X}_{1}^{\prime} \bm{X}_{1} \bm{\eta} \right) \notag \\
&+
E \left\{
\bm{\epsilon}_{0}^{\prime} \bm{X}_{0} \left( \bm{X}_{0}^{\prime} \bm{X}_{0} \right)^{-1} \bm{X}_{1}^{\prime} \bm{X}_{1} \left( \bm{X}_{0}^{\prime} \bm{X}_{0} \right)^{-1} \bm{X}_{0}^{\prime} \bm{\epsilon}_{0}
\right\}.
\end{align*}
After exchanging the trace and expectation operators, the right-hand side becomes
\begin{align*}
n_{1} \sigma_{1}^{2} + n_{1} k \sigma_{\eta}^{2} + \sigma_{0}^{2} \text{tr} \left\{ \left( \bm{X}_{0}^{\prime} \bm{X}_{0} \right)^{-1} \bm{X}_{1}^{\prime} \bm{X}_{1} \right\}.
\end{align*}
The second term is obtained because the predictors are standardized, i.e., $\frac{1}{n_{1}} \sum_{i \in \mathcal{S}_{1}} x_{1, il}^{2} = 1$ for all $l = 1, \dots, k$, where $x_{1, il}$ is the $l$-th predictor for observation $i$ in $\mathcal{S}_{1}$.
Dividing both sides by $n_{1}$ yields the result.
\end{proof}
We note a couple of points regarding to this result.
First, when the predictors are orthogonal to each other within each dataset, the third term reduces to $k \sigma_{0}^{2} /n_{0}$.
In this case, the prediction error converges to $\sigma_{1}^{2} + k \sigma_{\eta}^{2}$ as the sample size of $\mathcal{S}_{0}$, i.e., $n_{0}$, increases to infinity with the number of predictors fixed.
Second, the random coefficients assumption \ref{eq:random coefficients assumption} implies two datasets are similar when $\sigma_{\eta}^{2}$ is small.
This parameter contributes to the prediction error with the multiplier of $k$, the number of predictors.
Thus, even when $\sigma_{\eta}^{2}$ is small, the number of predictors matters.


\subsection{A generalization of the random coefficients model}
\label{subsec:A generalization of the random coefficients model}


The random coefficients model assumes that both datasets follow the same model but can differ in their marginal effects.
However, the model is also likely to vary across datasets.
In the context of linear regression, different models equate to different sets of predictors.
Corresponding to this situation, we have a generalization of the random coefficients model, which is given by
\begin{align}
\begin{cases}
\bm{y}_{0} = \bm{Z}_{0} \bm{\alpha}_{0} + \bm{X}_{0} \bm{\beta}_{0} + \bm{\epsilon}_{0}, \\
\bm{y}_{1} = \bm{W}_{1} \bm{\alpha}_{1} + \bm{X}_{1} \bm{\beta}_{1} + \bm{\epsilon}_{1}.
\end{cases}
\label{eq:RCM slide}
\end{align}
As in the previous model, $\bm{X}_{0}$ and $\bm{X}_{1}$ have the same predictors in their corresponding columns, and the homoskedastic random coefficients structure is retained in $\bm{\beta}_{j}$s.
The error term also has the same mean and variance as in the random coefficient model.
The remaining terms, $\bm{Z}_{0}$ and $\bm{W}_{1}$, may or may not share predictors, and we do not impose any specific structure on their regression coefficients $\bm{\alpha}_{j}$s.

The prediction of $\bm{y}_{1}$ in this model is given by the least squares estimate of $\bm{\beta}_{0}$ and $\bm{\alpha}_{1}$.
More precisely,
\begin{align*}
\hat{\bm{y}}_{1 \mid 0}
&=
\bm{W}_{1} \hat{\bm{\alpha}}_{1} + \bm{X}_{1} \hat{\bm{\beta}}_{0}, &&
\intertext{where}
\hat{\bm{\beta}}_{0} &= \left( \tilde{\bm{X}}_{0}^{\prime} \tilde{\bm{X}}_{0} \right)^{-1} \tilde{\bm{X}}_{0}^{\prime} \bm{y}_{0},
&
\hat{\bm{\alpha}}_{1} &= \left( \bm{W}_{1}^{\prime} \bm{W}_{1} \right)^{-1} \bm{W}_{1}^{\prime} \left( \bm{y}_{1} - \bm{X}_{1} \hat{\bm{\beta}}_{0} \right), \\
\tilde{\bm{X}}_{0} &= \bm{M}_{0} \bm{X}_{0},
&
\bm{M}_{0} &= \bm{I} - \bm{Z}_{0} \left( \bm{Z}_{0}^{\prime} \bm{Z}_{0} \right)^{-1} \bm{Z}_{0}^{\prime}.
\end{align*}
The matrix $\bm{M}_{0}$, which is called the annihilator in econometrics, is the difference between the unit matrix and the projection matrix.
The above result is a direct application of the Frisch-Waugh theorem (see \citet{frisch-waugh-33} for this theorem).

Further calculation leads to the following proposition, which is a generalization of the previous one.
\begin{prop}
The expectation of the average squared prediction error under the model and prediction described above is given by
\begin{align*}
\frac{1}{n_{1}} E \left( \bm{y}_{1} - \hat{\bm{y}}_{1 \mid 0} \right)^{\prime} \left( \bm{y}_{1} - \hat{\bm{y}}_{1 \mid 0} \right)
=
&\frac{n_{1} - k_{1}}{n_{1}} \sigma_{1}^{2}
+
\frac{\sigma_{\eta}^{2}}{n_{1}} \text{tr} \left( \tilde{\bm{X}}_{1}^{\prime} \tilde{\bm{X}}_{1} \right) \notag \\
&+
\frac{\sigma_{0}^{2}}{n_{1}} \text{tr} \left\{ \left( \tilde{\bm{X}}_{0}^{\prime} \tilde{\bm{X}}_{0} \right)^{-1} \tilde{\bm{X}}_{1}^{\prime} \tilde{\bm{X}}_{1} \right\},
\end{align*}
where $k_{1}$ is the number of predictors in $\bm{W}_{1}$, $\tilde{\bm{X}}_{1} = \bm{M}_{1} \bm{X}_{1}$, and $\bm{M}_{1} = \bm{I} - \bm{W}_{1} \left( \bm{W}_{1}^{\prime} \bm{W}_{1} \right)^{-1} \bm{W}_{1}^{\prime}$.
\label{prop:generalized}
\end{prop}
The proof is very similar to the previous one, and is thus left to Appendix \ref{app:A proof of Proposition 2}.
When the constant regressor is explicitly included in the model and is in $\bm{X}_{j}$s, $n_{1}$ and $\frac{n_{1}}{n_{0}}$ are added to the first and second trace terms, respectively.
In this paper, we always include the constant regressor in the common predictors ($X_{j}$s).

As in the first comment on Proposition \ref{prop:rcm}, the second and third terms become $(k - k_{1}) \sigma_{\eta}$ and $(k - k_{1}) \sigma_{0}^{2} / n_{0}$, respectively, when predictors are orthogonal to each other.
It is reasonable because $k - k_{1}$ predictors are common between two regression models \eqref{eq:RCM slide} and their influences are included in the second and third terms of the prediction error .
Then, the prediction error is reduced to $(n_{1} - k_{1}) \sigma_{1}^{2} / n_{1} + (k - k_{1}) \sigma_{\eta} + (k - k_{1}) \sigma_{0}^{2} / n_{0}$, and converges to $\sigma_{1}^{2} + (k - k_{1}) \sigma_{\eta}$ as the sample sizes of $\mathcal{S}_{0}$ and $\mathcal{S}_{1}$ increase to infinity keeping the number of predictors fixed.


\section{Estimation of the prediction error}
\label{sec:Estimate the prediction error}


Three parameters ($\sigma_{j}^{2}$s, $\sigma_{\eta}^{2}$) need to be established in order to estimate the prediction error.
We take the Bayesian approach to their estimation because it is flexible and straightforward to calculate prediction error estimates for our random coefficients models in a hierarchical manner and these estimates can be derived under the squared loss (see, e.g., \citet{berger-85} for other desired properties of the Bayesian approach).


\subsection{Estimation of $\sigma_{j}^{2}$}

The model to be used for estimation is either Equation \eqref{eq:RCM common} or \eqref{eq:RCM slide}.
In addition to this model, we assume error terms to follow the normal distribution with specified mean and variance.
Then, it becomes the normal linear regression model but the dataset for estimation is restricted to $\mathcal{S}_{j}$.

For regression coefficients except for the intercept, we assume the hyper-$g$ prior, setting its hyperparameter as $a = 3$ (see \citet{liang-etal-08}).
Remaining parameters are assumed to follow noninformative prior distributions: the prior for the intercept is proportional to a constant, and the priors for the variances of the regression errors are $\pi ( \sigma_{j}^{2} ) \propto \sigma_{j}^{-2}$ ($j = 0, 1$).
With this prior specification, posterior means for model parameters are analytically tractable (see, e.g., \citet{miyawaki-maceachern-19}).
Then, the variances of the regression errors are estimated as their corresponding posterior means.


\subsection{Estimation of $\sigma_{\eta}^{2}$}

There are two approaches to estimate the variance of the random coefficients structure, $\sigma_{\eta}^{2}$.
The first is to estimate it as the sample variance of the posterior means of the regression coefficients described in the previous subsection without correcting for the degrees of freedom.
The other is to specify an additional prior on $\sigma_{\eta}^{2}$ and use its posterior mean as its estimate.
Although the latter is straightforward from the theoretical view point, we recommend the former approach because of its simplicity and less computational burden.
Comparison of these approaches will be discussed in the next subsection.



\subsection{Comparison of two approaches}
\label{subsec:Evaluation of hat{sigma}_{eta}^{2}}


This subsection evaluates how the estimation approach affects the prediction error.
First, we describe the approach that specifies standard prior distributions on all model parameters.
We assume the half standard Cauchy prior distribution on $\sigma_{\eta}$ (see, e.g., \citet{gelman-06} for a discussion of this prior specification).
The other priors remain the same.
Because the posterior means under this prior setting are not analytically tractable, we use the Markov chain Monte Carlo method to draw samples from the posterior and estimate the posterior means.
The full conditional distributions are derived in Appendix \ref{app:Full conditional distributions for the random coefficients model}.

Next, we compare these two approaches by using the diabetes data.
The approximated approach is our recommendation, while the approach with the Cauchy prior is the one described above.
The following comparison is performed model by model, and we restrict our attention to models in which the number of common predictors (not including the constant regressor) is greater than seven.
The resulting number of models is 436.

Figure \ref{fig:Prediction error based on the approximated approach and standardized deviance.} is the scatter plot of prediction errors based on the approximated approach and standardized absolute deviances of prediction errors between two approaches.
\begin{figure}[ht]
\centering
\includegraphics[width=12cm, clip, keepaspectratio]{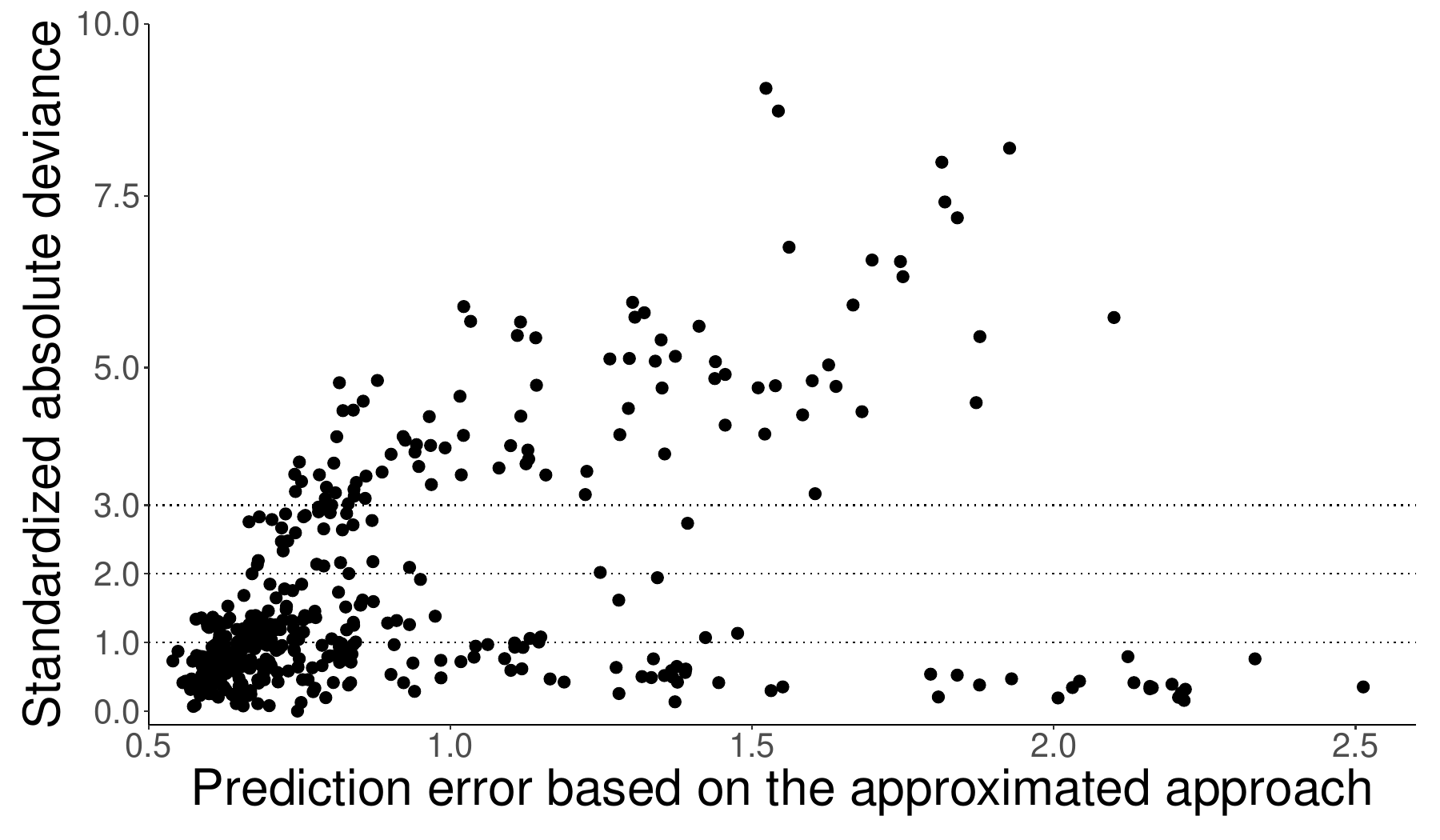}
\caption{Prediction errors based on the approximated approach and standardized deviances.}
\label{fig:Prediction error based on the approximated approach and standardized deviance.}
\end{figure}
While the approach with the Cauchy prior produces samples from the marginal posterior distribution of the prediction error, the approximated approach yields its point estimate.
We take the posterior mean for the former as its point estimate.
Then, the absolute deviance between these two is a measure of difference between these approaches.
Further, to be comparable between models, this deviance is standardized by the posterior standard deviation from the approach with the Cauchy prior, which is termed as the standardized absolute deviance in this figure.

When prediction errors are small (which is mainly of our interest), this deviance is small, and vice versa.
This result suggests that the approximated approach performs well when its prediction errors are small at least with the diabetes data.

The deviance can be decomposed into three components: deviances of $\sigma_{\eta}^{2}, \sigma_{0}^{2}, \sigma_{1}^{2}$.
They are further examined in Figure \ref{fig:Scatter plots of the two variance estimates.} and Table \ref{table:Difference between the alternative Bayesian and approximation approaches}.
The left panel of Figure \ref{fig:Scatter plots of the two variance estimates.} is the scatter plot of estimated $\sigma_{\eta}^{2}$ based on two approaches and the second row of Table \ref{table:Difference between the alternative Bayesian and approximation approaches} gives summary statistics of difference of the two.
\begin{figure}[ht]
\centering
\includegraphics[width=14cm, clip, keepaspectratio]{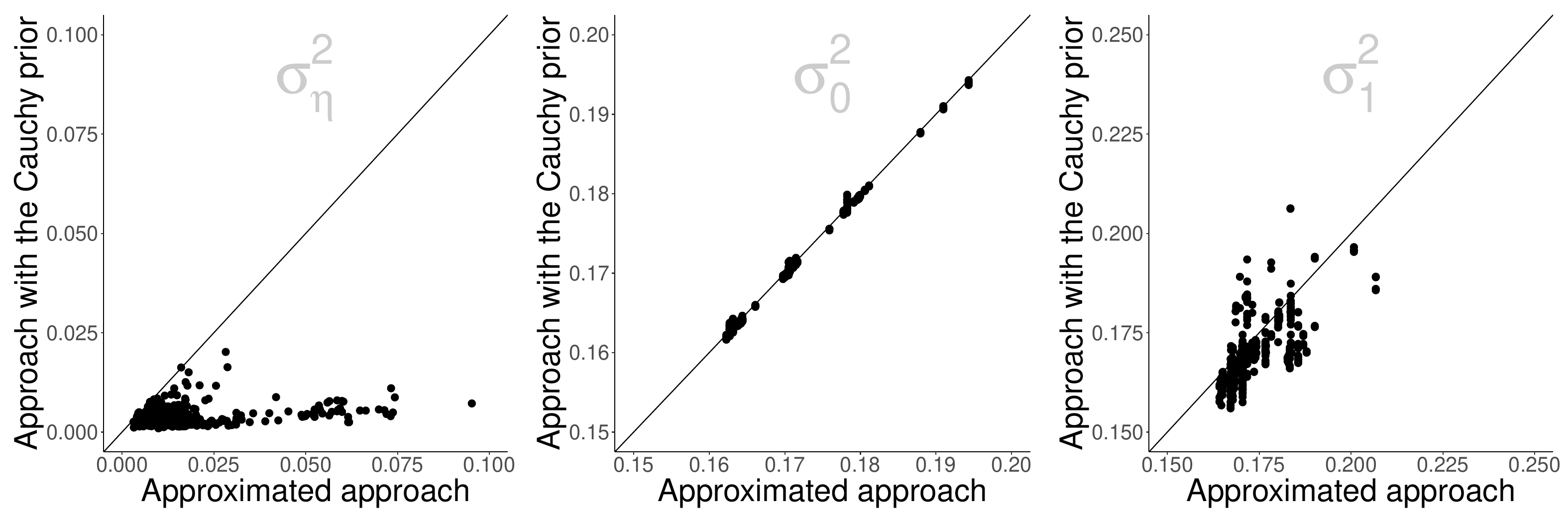}
\caption{Scatter plots of the two variance estimates.}
\label{fig:Scatter plots of the two variance estimates.}
\end{figure}
\begin{table}[!ht]
\begin{center}
\catcode`?=\active \def?{\phantom{0}}
\catcode`!=\active \def!{\phantom{-}}
\caption{Difference between the approach with the Cauchy prior and the approximated approach}
\label{table:Difference between the alternative Bayesian and approximation approaches}
\scalebox{0.8}{
\begin{tabular}{@{}lcccccc}
\toprule
         & Minimum & 1st Quartile & Median & Mean (SD) & 3rd Quartile & Maximum \\
\midrule
$\hat{\sigma}_{\eta}^{2}$ & $-0.088??$ & $-0.014??$ & $-0.0063?$ & $-0.012??$ ($0.015??$) & $-0.0035?$ & $0.00040$ \\
$\hat{\sigma}_{0}^{2}$    & $-0.00078$ & $-0.00037$ & $-0.00025$ & $-0.00015$ ($0.00037$) & $-0.00011$ & $0.0016?$ \\
$\hat{\sigma}_{1}^{2}$    & $-0.021??$ & $-0.0066?$ & $-0.0038?$ & $-0.0042?$ ($0.0063?$) & $-0.0018?$ & $0.023??$ \\
\bottomrule
\end{tabular}
}
\end{center}
\end{table}%
The estimated $\sigma_{\eta}^{2}$s based on the approximated approach are larger than those estimated from the approach with the Cauchy prior because most dots are below the 45 degree line.
In contrast, $\sigma_{j}^{2}$s ($j = 0, 1$) are not affected in this manner because dots are around the 45 degree line (see the middle and right panels of above figure and the last two rows of this table).

Combining above results, the difference of two approaches shown in Figure \ref{fig:Prediction error based on the approximated approach and standardized deviance.} mainly comes from the difference of estimates of $\sigma_{\eta}^{2}$.
We consider this is partly due to the prior specification and its further investigation is our future work.

The approach with the Cauchy prior has two disadvantages.
First, it is computationally intensive, because we need to use the Markov chain Monte Carlo method to estimate the posterior means, which, for this subset of models, takes approximately 20 minutes with the Fortran language and a 3.4 GHz Intel Core i7 computer.
The number of models grows rapidly when we investigate the heterogeneity/similarity in the set of predictors (see Subsection \ref{subsec:Similarity search}), which makes this specification infeasible.
Second, the estimation is sensitive to the prior choice when the number of common predictors is small.
When only the constant regressor is common, as an extreme case, the full conditional for $\sigma_{\eta}^{2}$ does not have a mean, which produces a larger posterior mean.
Future research will involve a search for a prior that is both robust to the number of predictors and simple for estimation.


\section{Illustrative applications}
\label{sec:Illustrative applications}


\subsection{Internal data compatibility}


This subsection focuses on the problem raised by data quality and applies Proposition \ref{prop:rcm} to the geyser data.
Because they are repeated observations, it is reasonable to assume that they are from the same population.
Then, the primary interest of this application is to decide to pool or not to pool two datasets to estimate the model parameters.
The response is the logarithm of waiting time at time $t+1$, and two potential predictors are considered: the duration at time $t$ (denoted by $x_{1}$), and an indicator variable that is one if the duration is less than or equal to 2.5 and is zero otherwise (denoted by $x_{2}$).

Table \ref{table:Internal quality of the geyser data} presents the prediction errors estimated with the approximated approach when changing a set of predictors.
\begin{table}[!ht]
\begin{center}
\catcode`?=\active \def?{\phantom{0}}
\catcode`!=\active \def!{\phantom{.}}
\caption{Internal quality of the geyser data}
\label{table:Internal quality of the geyser data}
\begin{tabular}{@{}lcccc}
\toprule
Predictors       & \multicolumn{4}{c}{$\mathcal{S}_{0} \mid \mathcal{S}_{1}$} \\ \cline{2-5}
                 & Term 1 & Term 2 & Term 3 & Prediction error \\
\midrule
$\emptyset$      & 0.032? & 0!????? & 0.00021 & 0.032? \\
$x_{1}$          & 0.0068 & 0.0020? & 0.00008 & 0.0089 \\
$x_{2}$          & 0.0074 & 0.00005 & 0.00009 & 0.0076 \\
$(x_{1}, x_{2})$ & 0.0069 & 0.00089 & 0.00008 & 0.0079 \\
\midrule
Predictors       & \multicolumn{4}{c}{$\mathcal{S}_{1} \mid \mathcal{S}_{0}$} \\ \cline{2-5}
                 & Term 1 & Term 2 & Term 3 & Prediction error \\
\midrule
$\emptyset$      & 0.046? & 0!????? & 0.00041 & 0.047?? \\
$x_{1}$          & 0.0087 & 0.0020? & 0.00018 & 0.010?? \\
$x_{2}$          & 0.010? & 0.00005 & 0.00019 & 0.010?? \\
$(x_{1}, x_{2})$ & 0.0081 & 0.00089 & 0.00055 & 0.0096? \\%
\bottomrule
\end{tabular}
\end{center}
\end{table}
The first column indicates which predictors are included in the model, noting that the constant regressor is always included.
The last column gives the estimated prediction error, which is the sum of the preceding three columns (the second to fourth columns), as presented in Proposition \ref{prop:rcm}.

When the same predictors are given, two prediction errors are close (for example, 0.0089 and 0.010 for the model with $x_{1}$ only), suggesting that it is reasonable to pool the entire set of data to estimate the model parameters of interest.
We note a couple of points regarding above result.
First, that the prediction errors of $\mathcal{S}_{0} \mid \mathcal{S}_{1}$ are smaller than those of $\mathcal{S}_{1} \mid \mathcal{S}_{0}$ suggests that precise data are useful for developing better models at the cost of collecting high-quality data.
Second, the binary variable ($x_{2}$) is comparable to the continuous variable ($x_{1}$) in terms of prediction, indicating that the rounded measurements at night (2 or 4) are sufficient for prediction based on linear regression.


\subsection{Similarity search}
\label{subsec:Similarity search}


This subsection considers the two-dataset problem characterized by different data sources.
Proposition \ref{prop:generalized} can be used to evaluate the heterogeneity/similarity in the set of predictors.

In general, when we have two potential model specifications, the model with the smaller prediction error is the better model, in terms of prediction.
In the two-dataset problem, we obtain two prediction errors by changing $0$ and $1$, i.e., $\frac{1}{n_{1}} E ( \bm{y}_{1} - \hat{\bm{y}}_{1 \mid 0} )^{\prime} ( \bm{y}_{1} - \hat{\bm{y}}_{1 \mid 0} )$ and $\frac{1}{n_{0}} E ( \bm{y}_{0} - \hat{\bm{y}}_{0 \mid 1} )^{\prime} ( \bm{y}_{0} - \hat{\bm{y}}_{0 \mid 1} )$.
It is therefore reasonable to select the model with the smallest sum of estimates of these prediction errors.
This process is a specific form of cross-validation when two datasets are drawn from the identical population.
In the two-dataset problem, however, we allow for the datasets to be drawn from potentially different populations.
Even when this is the case, a smaller prediction error based on the generalized random coefficients model indicates a better model.

When we know one of two datasets is more appropriate to choose a set of predictors in advance, it is possible to use a weighted sum of estimates instead of the equal sum of them.
Because of its subjectivity, we do not pursue this specification in the following application.
Furthermore, the model with the smallest maximum can be selected, though we do not take this specification in this paper.

One application of this approach is to find common predictors across datasets.
Assume first that one predictor is common across datasets in the state of nature.
We then have two specifications: one with the predictor in the model and the other without the predictor.
A smaller prediction error in the former is an indication that the predictor is common across datasets and a support for the assumption.
When several specifications share predictors of interest (or several specifications do not share them), a possible approach is to take average potential models, which is reasonable under the squared loss function (see also \citet{miyawaki-maceachern-19}).

Because the Bayesian approach is cohesive framework for estimating model parameters and model uncertainty, we use it to estimate the prediction errors and average them with the posterior model probability, which is proportional to the product of the marginal likelihood and prior model probability given a model.
When the uniform prior model probability is used, as in this paper, the result is proportional to the marginal likelihood.
In this application, because the prediction error is averaged over the posterior model probability (a Bayesian measure of in-sample fit), it balances the measure of prediction and fit.

The selection process in this context is summarized in the following five steps.
\begin{description}
\item[Step 1.] Select a set of predictors assumed to be common across datasets.
\item[Step 2.] Given the models represented by Equation \eqref{eq:RCM slide} (one for $\mathcal{S}_{0}$ and the other for $\mathcal{S}_{1}$), estimate the prediction error.
\item[Step 3.] Average the error over potential models based on their posterior model probabilities.
\item[Step 4.] Repeat the three steps above by changing the labels $0$ and $1$; the prediction error for this set of common predictors is the sum of the two errors.
\item[Step 5.] Choose the set of predictors that minimizes the prediction error as common predictors across datasets.
\end{description}

The diabetes data are used to illustrate the application of our selection of common predictors.
The response is the log of the diabetes progression measure, and the ten predictors explained in Section \ref{sec:Motivating data} are used.
The maximum marginal likelihood criterion indicates that $x_{2}$, $x_{3}$, and $x_{9}$ are common across datasets (see Section \ref{sec:Motivating data}).
\citet{miyawaki-maceachern-19} report that $x_{2}$, $x_{3}$, $x_{4}$, $x_{5}$, $x_{6}$, and $x_{9}$ give the least mean squared loss when using the same model across all observations.

We performed our selection process as described above.
The selection based on the prediction error is summarized in the left panel of Figure \ref{fig:Similarity selection with diabetes data.}.
\begin{figure}[ht]
\begin{minipage}[b]{.48\linewidth}
\centering
\includegraphics[width=7cm, clip, keepaspectratio]{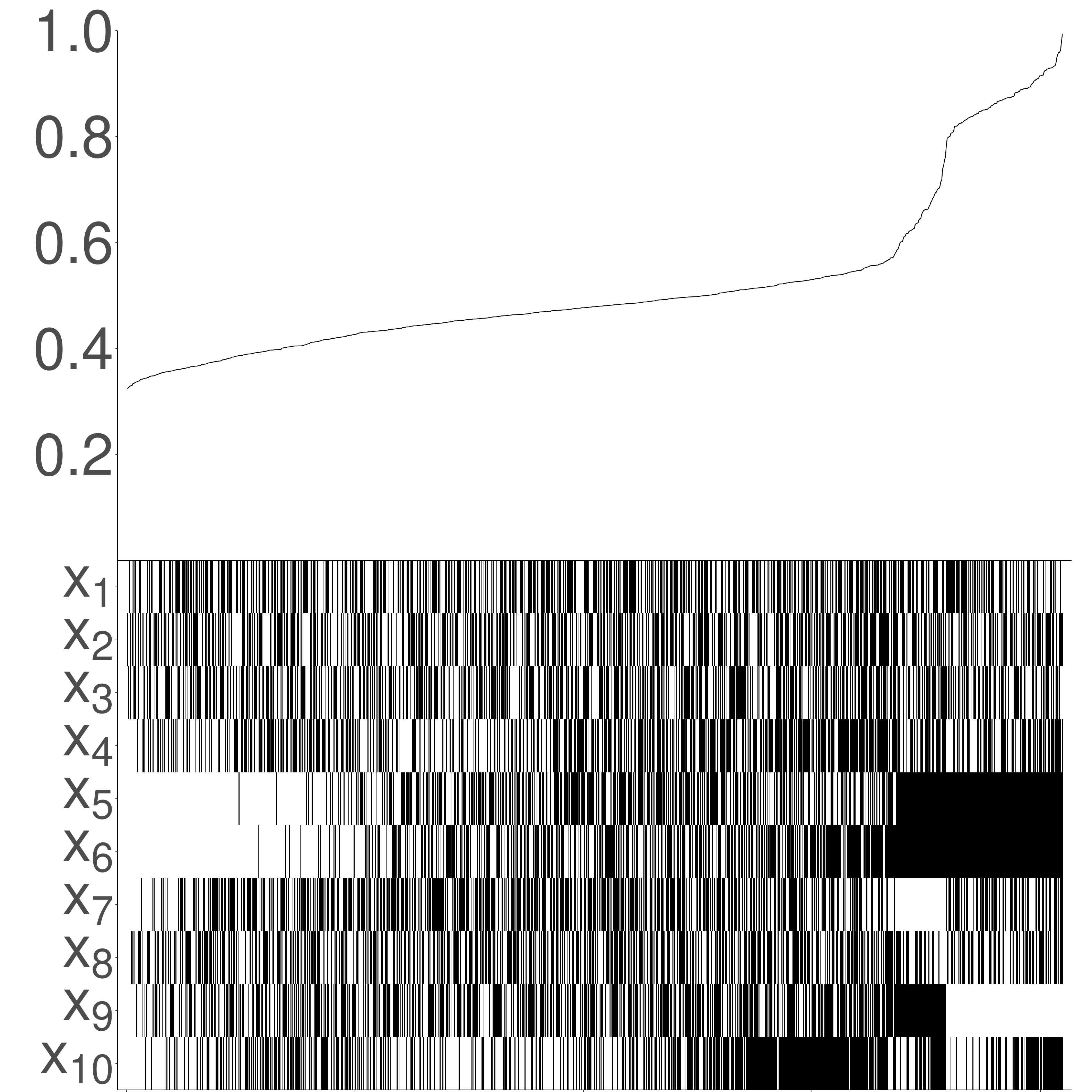}
\subcaption{All selections.}
\end{minipage}
\begin{minipage}[b]{.48\linewidth}
\centering
\includegraphics[width=7cm, clip, keepaspectratio]{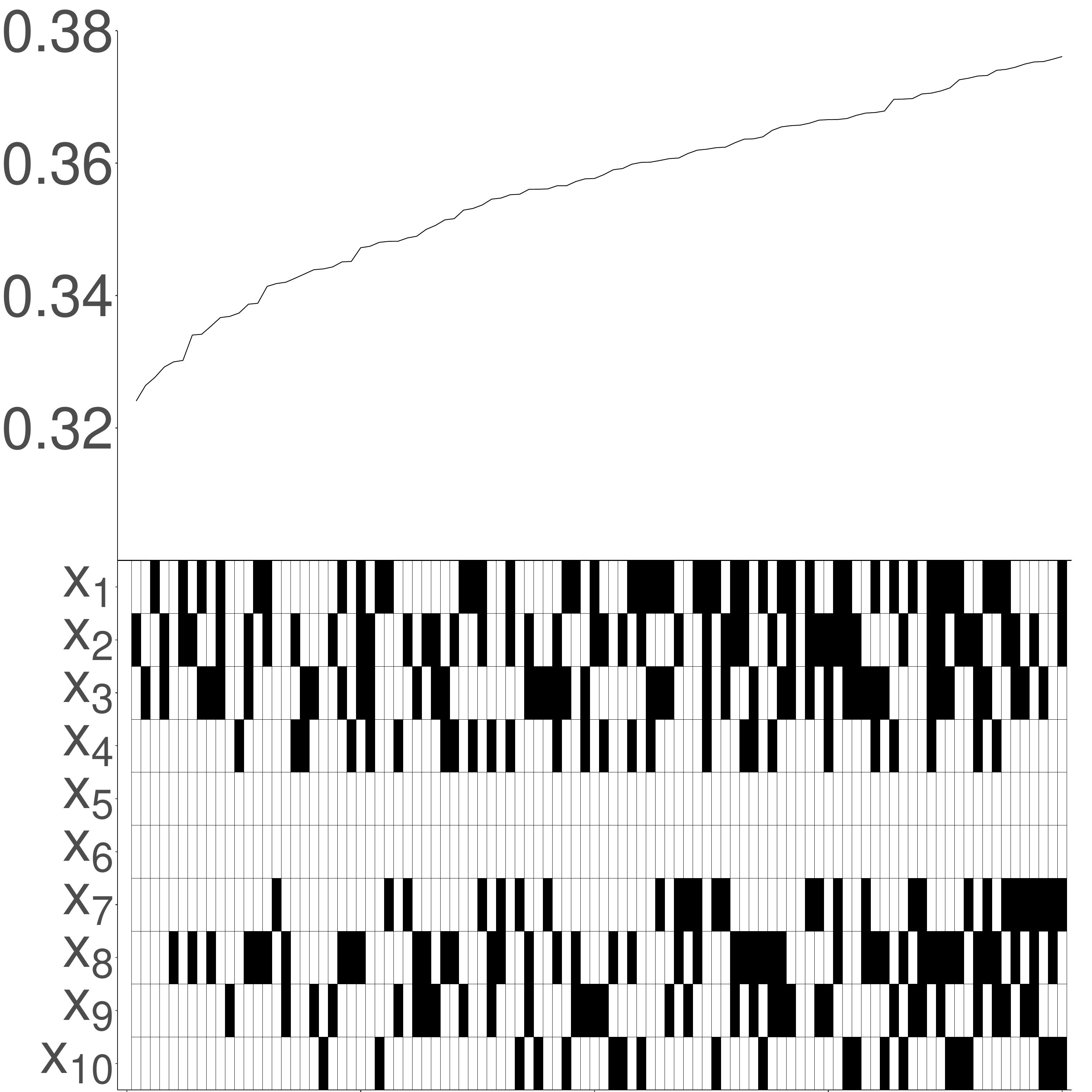}
\subcaption{Top 100 selections.}
\end{minipage}
\caption{Selection of common predictors with the diabetes data.}
\label{fig:Similarity selection with diabetes data.}
\end{figure}
The top plot draws the prediction error in ascending order, and the bottom map shows the corresponding common predictors.
When the row for a predictor---for example, $x_{1}$---is black, it is selected as a common predictor.

We note two findings observed in this figure.
First, the prediction error increases rapidly at the right side of the figure.
This increase appears to be closely related to the inclusion of $x_{5}$ and $x_{6}$ as common predictors.
Further, they are not included as common predictors in the top one hundred models, as found in the right panel of Figure \ref{fig:Similarity selection with diabetes data.}.
Because their measurement accuracy seems to be similar across datasets, this rise in prediction error attributes to the difference of their association to the response across datasets (See also Figure \ref{fig:Decomposition.}).
The smallest prediction error is achieved with only $x_{2}$ as a common predictor.

The prediction error is decomposed into the three terms in Proposition \ref{prop:generalized}, as shown in Figure \ref{fig:Decomposition.}.
\begin{figure}[ht]
\centering
\includegraphics[width=14cm, clip, keepaspectratio]{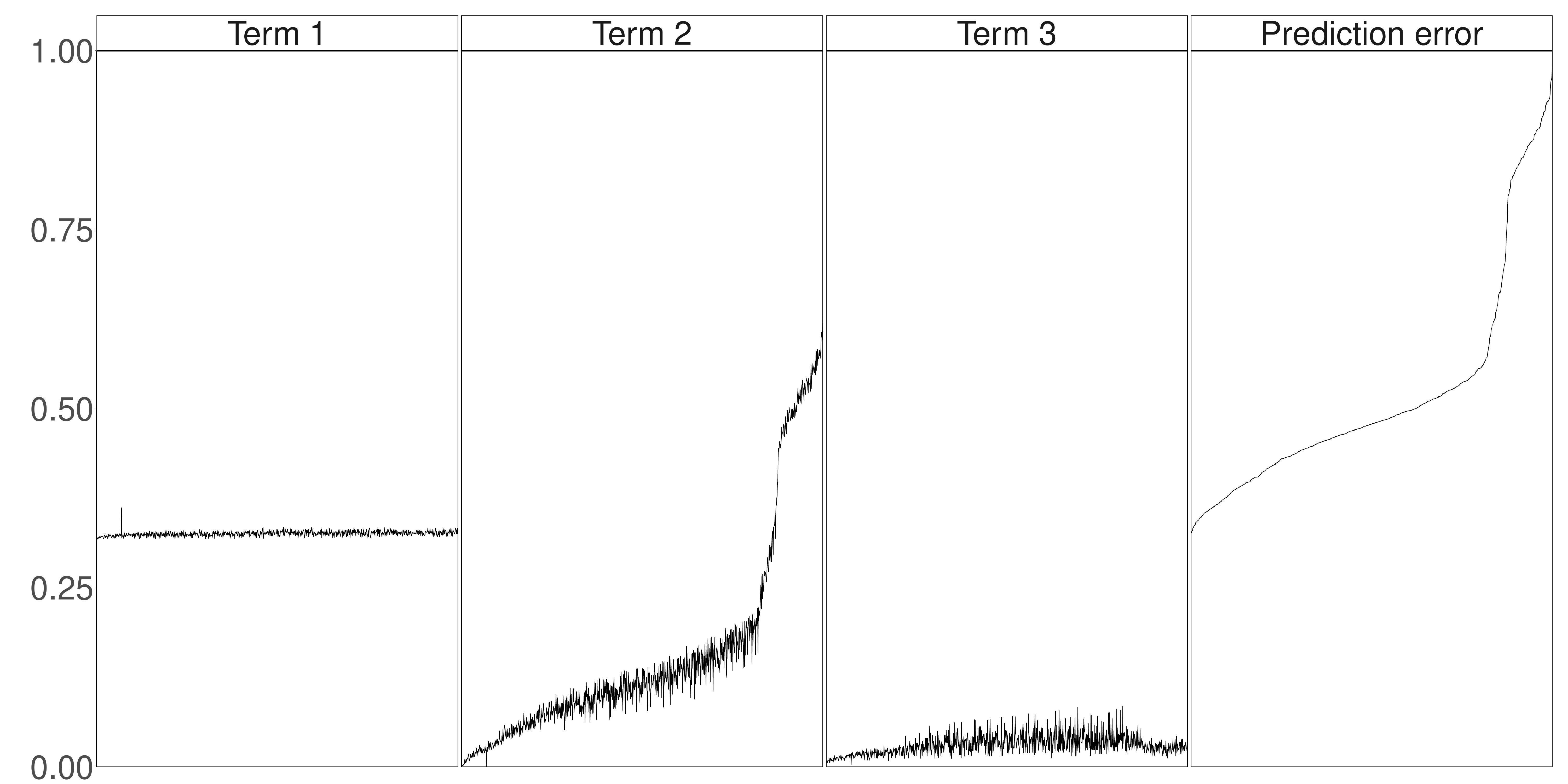}
\caption{Decomposition of the prediction error.}
\label{fig:Decomposition.}
\end{figure}
The sharp increase in the prediction error is mainly due to the second term.

Regression analysis decomposes the association between the response and predictors into two parts: the mean response and the error term.
When this association is linear and the error is distributed as normal, they reduce to coefficients and the error variance, respectively.
Two datasets are different when these two components are different.
Our three-term decomposition illustrates this view.
The second term captures the difference in coefficients, while remaining terms consider the heterogeneity in the error term variances.
Above figure shows the difference in coefficients becomes larger as the prediction error increases.
Although this result is for a specific dataset, the random coefficients structure is useful for model selection.


\section{Conclusion}
\label{sec:Conclusion}


When data are obtained from potentially different sources, the issue of whether we should pool the data to estimate the parameters of interest is an important one.
One conservative choice is to separate the data and perform estimation using different models for different datasets, at the cost of inaccurate estimates due to reduced information from the subset of data.
As we found for the geyser dataset, the difference between the two datasets can be small in terms of prediction.
However, in general, the difference can be large, so that different models are required.
The second prediction error or other variable selection method such as the one proposed by \citet{miyawaki-maceachern-19} is useful for finding common predictors to aid the model building process.
The former is demonstrated with the diabetes dataset.

A possible empirical application of our method is the determinant of growth, which is an important question in economics that has been analyzed in a number of papers (see, e.g., \citet{fernandez-etal-01} and references therein).
However, the conclusion may change when we take the two-dataset problem into account.
Economic theory typically assumes one common mechanism of growth.
This assumption may be true, but heterogeneity in the set of predictors exists.
For example, data from developed and developing counties are different in terms of accuracy.
Our approach is useful in this case.

We would like to note two future issues to address.
The first is the extension of the case, where we have more than two datasets.
In this case, the computational burden would increase by a power of more than two, making our estimation infeasible.
Improvement in the methodology and/or computation would be necessary to reduce the computational burden.

The second issue is the prior on $\sigma_{\eta}^{2}$.
For its square root, the half Cauchy distribution is used in this paper, but this specification requires the Markov chain Monte Carlo sampling.
To make our method scalable, a reasonable prior with a shorter computational time must be developed.


\appendix
\section{Proof of Proposition \ref{prop:generalized}}
\label{app:A proof of Proposition 2}


The prediction error is given by
\begin{align*}
\bm{y}_{1} - \hat{\bm{y}}_{1 \mid 0}
=
\bm{y}_{1} - \bm{Z}_{1} \hat{\bm{\alpha}}_{1} - \bm{X}_{1} \hat{\bm{\beta}}_{0}
=
\bm{M}_{1} \bm{y}_{1} - \bm{M}_{1} \bm{X}_{1} \hat{\bm{\beta}}_{0}
=
\bm{M}_{1} \left( \bm{y}_{1} - \bm{Q} \bm{y}_{0} \right),
\end{align*}
where $\bm{Q} = \bm{X}_{1} ( \tilde{\bm{X}}_{0}^{\prime} \tilde{\bm{X}}_{0} )^{-1} \tilde{\bm{X}}_{0}^{\prime}$.
Linear models for $\bm{y}_{j}$ ($j = 0, 1$) and the random coefficients structure give the right-hand side as
\begin{align*}
\bm{M}_{1} \left( \bm{\epsilon}_{1} + \bm{X}_{1} \bm{\eta} - \bm{Q} \bm{\epsilon}_{0} \right).
\end{align*}
We note that $\bm{Q} \bm{X}_{0} = \bm{X}_{1}$ and $\bm{Q} \bm{Z}_{0} = \bm{O}$, where $\bm{O}$ is the null matrix.
The independence between errors gives the expected sum of squared prediction errors as
\begin{align*}
E \left( \bm{y}_{1} - \hat{\bm{y}}_{1 \mid 0}\right)^{\prime} \left( \bm{y}_{1} - \hat{\bm{y}}_{1 \mid 0}\right)
=
\left( n_{1} - k_{1} \right) \sigma_{1}^{2}
+
\sigma_{\eta}^{2} \text{tr} \left( \bm{X}_{1}^{\prime} \bm{M}_{1} \bm{X}_{1} \right)
+
\sigma_{0}^{2} \text{tr} \left( \bm{Q}^{\prime} \bm{M}_{1} \bm{Q} \right).
\end{align*}
Because the annihilator is idempotent, the first trace term is $\text{tr} ( \tilde{\bm{X}}_{1}^{\prime} \tilde{\bm{X}}_{1} )$.
The second trace term then becomes
\begin{align*}
\text{tr} \left( \bm{Q}^{\prime} \bm{M}_{1} \bm{Q} \right)
&=
\text{tr} \left\{ \left( \tilde{\bm{X}}_{0}^{\prime} \tilde{\bm{X}}_{0} \right)^{-1} \bm{X}_{1}^{\prime} \bm{X}_{1} \right\}
-
\text{tr} \left\{ \left( \tilde{\bm{X}}_{0}^{\prime} \tilde{\bm{X}}_{0} \right)^{-1} \bm{X}_{1}^{\prime} \bm{P}_{1} \bm{X}_{1} \right\} \notag \\
&=
\text{tr} \left\{ \left( \tilde{\bm{X}}_{0}^{\prime} \tilde{\bm{X}}_{0} \right)^{-1} \tilde{\bm{X}}_{1}^{\prime} \tilde{\bm{X}}_{1} \right\},
\end{align*}
where $\bm{P}_{1} = \bm{W}_{1} \left( \bm{W}_{1}^{\prime} \bm{W}_{1} \right)^{-1} \bm{W}_{1}^{\prime}$.
Collecting terms and dividing both sides by $n_{1}$ yields the result.


\section{Full conditional distributions for the random coefficients model}
\label{app:Full conditional distributions for the random coefficients model}


The model is represented by Equation \eqref{eq:RCM slide} and the random coefficients structure.
Because we always include the intercept and assume the random coefficients structure on it,  the model specification in this context is given by
\begin{align*}
\bm{y}_{0} &= \bm{Z}_{0} \bm{\alpha}_{0} + \xi_{0} + \bm{X}_{0} \bm{\beta}_{0} + \bm{\epsilon}_{0}, \\
\bm{y}_{1} &= \bm{Z}_{1} \bm{\alpha}_{1} + \xi_{1} + \bm{X}_{1} \bm{\beta}_{1} + \bm{\epsilon}_{1}, \\
\begin{pmatrix}
\xi_{1} \\ \bm{\beta}_{1}
\end{pmatrix}
&=
\begin{pmatrix}
\xi_{0} \\ \bm{\beta}_{0}
\end{pmatrix}
+
\bm{\eta}.
\end{align*}
For brevity, we replace $\bm{W}_{1}$ with $\bm{Z}_{1}$ in this appendix.
Let
\begin{align*}
\bm{\Lambda}_{j}^{-1}
=
\begin{pmatrix}
\bm{Z}_{j} \vdots \bm{X}_{j}
\end{pmatrix}^{\prime}
\begin{pmatrix}
\bm{Z}_{j} \vdots \bm{X}_{j}
\end{pmatrix},
\quad
(j = 0, 1).
\end{align*}
The prior distributions are specified in Section \ref{sec:Estimate the prediction error}.

Then, the combination of the $g$-prior and the random coefficients structure results in the following conditional distribution,
\begin{align*}
\begin{pmatrix}
\xi_{j} \\ \bm{\alpha}_{j} \\ \bm{\beta}_{j}
\end{pmatrix}
\mid
\xi_{1-j}, \bm{\beta}_{1-j}
&\sim
N \left\{
\begin{pmatrix}
\xi_{1-j}
\\
\bm{\Sigma}_{j}
\begin{pmatrix}
\bm{0} \\ \sigma_{\eta}^{-2} \bm{\beta}_{1-j}
\end{pmatrix}
\end{pmatrix},
\begin{pmatrix}
\sigma_{\eta}^{2} & \bm{0}^{\prime} \\ \bm{0} & \bm{\Sigma}_{j}
\end{pmatrix}
\right\},
\intertext{where $\bm{0}$ is a vector of zeros, $N (\bm{\mu}, \bm{\Sigma})$ denotes the multivariate normal distribution with mean vector $\bm{\mu}$ and covariance matrix $\bm{\Sigma}$, $g_{j}$ is the hyperparameter in the $g$-prior, and}
\bm{\Sigma}_{j}^{-1} &= g_{j}^{-1} \sigma_{j}^{-2} \bm{\Lambda}_{j}^{-1}
+
\begin{pmatrix}
\bm{O} & \bm{O} \\ \bm{O} & \sigma_{\eta}^{-2} \bm{I}
\end{pmatrix},
\end{align*}
for $j = 0, 1$.

\noindent\textbf{The full conditional for $\xi_{j}, \bm{\alpha}_{j}, \bm{\beta}_{j}$ ($j = 0, 1$).}

It is the multivariate normal distribution with mean $\bm{m}_{j}$ and covariance matrix $\bm{S}_{j}$; that is,
\begin{align*}
\begin{pmatrix}
\xi_{j} \\ \bm{\alpha}_{j} \\ \bm{\beta}_{j}
\end{pmatrix}
\mid
\xi_{1-j}, \bm{\beta}_{1-j}
&\sim
N \left( \bm{m}_{j}, \bm{S}_{j} \right),
\intertext{where}
\bm{S}_{j}^{-1}
&=
\begin{pmatrix}
\sigma_{\eta}^{-2} + \sigma_{j}^{-2} n_{j} & \bm{0}^{\prime} \\ \bm{0} & \tilde{\bm{\Sigma}}_{j}^{-1}
\end{pmatrix}, \\
\bm{m}_{j}
&=
\begin{pmatrix}
\left( \sigma_{\eta}^{-2} + \sigma_{j}^{-2} n_{j} \right)^{-1} \left( \sigma_{\eta}^{-2} \xi_{1-j} + \sigma_{j}^{-2} n_{j} \bar{y}_{j} \right) \\
\tilde{\bm{\Sigma}}_{j}
\begin{pmatrix}
\sigma_{j}^{-2} \bm{Z}_{j}^{\prime} \bm{y}_{j} \\ \sigma_{\eta}^{-2} \bm{\beta}_{1-j} + \sigma_{j}^{-2} \bm{X}_{j}^{\prime} \bm{y}_{j}
\end{pmatrix}
\end{pmatrix}, \\
\tilde{\bm{\Sigma}}_{j}^{-1}
&=
\left( 1 + g_{j}^{-1} \right) \sigma_{j}^{-2} \bm{\Lambda}_{j}^{-1}
+
\begin{pmatrix}
\bm{O} & \bm{O} \\ \bm{O} & \sigma_{\eta}^{-2} \bm{I}
\end{pmatrix}.
\end{align*}

\noindent\textbf{The full conditional for $\sigma_{j}^{2}$ ($j = 0, 1$).}

It is the inverse gamma distribution; that is,
\begin{align*}
\sigma_{j}^{2} \mid \xi_{j}, \bm{\alpha}_{j}, \bm{\beta}_{j}
\sim
IG \left( \frac{n_{j}}{2}, \frac{\bm{e}_{j}^{\prime} \bm{e}_{j}}{2} \right),
\end{align*}

where $IG (a, b)$ is the inverse gamma distribution with shape parameter $a$ and scale parameter $b$ and $\bm{e}_{j} = \bm{y}_{j} - \bm{Z}_{j} \bm{\alpha}_{j} - \xi_{j} - \bm{X}_{j} \bm{\beta}_{j}$.

\noindent\textbf{The full conditional for $g_{j}$ ($j = 0, 1$).}

It is a nonstandard distribution with a conditional density function proportional to
\begin{align*}
g_{j}^{-k/2} \left( 1 + g_{j} \right)^{-a/2} \exp \left( -\frac{ \bm{\beta}_{j}^{\prime} \bm{X}_{j}^{\prime} \bm{X}_{j} \bm{\beta}_{j} }{2 \sigma_{j}^{2}} \right),
\end{align*}
where $k$ is the number of predictors in $\bm{X}_{j}$ and $a$ is the parameter in the prior distribution on $g_{j}$, which is equal to $3$ in our applications.
The Metropolis-Hasting algorithm with the inverse gamma proposal is used to draw a sample from this density.
The proposal distribution is specified as
\begin{align*}
IG \left( \frac{k}{2}, \frac{ \bm{\beta}_{j}^{\prime} \bm{X}_{j}^{\prime} \bm{X}_{j} \bm{\beta}_{j} }{2 \sigma_{j}^{2}} \right).
\end{align*}
%

\noindent\textbf{The full conditionals for $\sigma_{\eta}^{2}$ and the auxiliary variable.}

As shown in \citet{wand-etal-11} and \citet{makalic-schmidt-16}, the half standard Cauchy distribution can be represented as a mixture of two inverse gamma distributions: when $U \sim IG ( 1/2, 1 )$ and $\sigma_{\eta}^{2} \mid U \sim IG ( 1/2, U^{-1} )$, $\sigma_{\eta}$ follows the half standard Cauchy distribution.

By means of this fact, the generation of $\sigma_{\eta}^{2}$ is augmented by the auxiliary variable $U$.
Their full conditionals are $IG ( \frac{k+2}{2}, \frac{S_{\eta}}{2} + U^{-1} )$, where $S_{\eta} = ( \xi_{1} - \xi_{0} )^{2} + ( \bm{\beta}_{1} - \bm{\beta}_{0} )^{\prime} ( \bm{\beta}_{1} - \bm{\beta}_{0} )$ for $\sigma_{\eta}^{2}$ and $IG ( 1, 1 + \sigma_{\eta}^{-2} )$ for $U$.


\section*{Acknowledgements}

One of the authors acknowledges that this paper is supported by the scholarship from the Kwansei Gakuin University.


\bibliographystyle{./chicago}
\bibliography{./two_bib}


\end{document}